\documentclass[10pt,rqno, a4paper]{article}

\usepackage{titlesec}
\titleformat*{\section}{\bf\centering} 
\titleformat*{\subsection}{\centering\bf} 
\titleformat*{\subsubsection}{\centering\it} 

\usepackage{amsmath}
\usepackage{amssymb}
\usepackage{amsfonts}
\usepackage{setspace}
\usepackage{latexsym}
\usepackage{mathrsfs}
\usepackage{epsfig}
\usepackage{amsthm}
\usepackage{color}

\usepackage{enumitem}

\usepackage[margin=2cm]{caption}

\newcommand{\be}{\begin{equation}}
\newcommand{\ee}{\end{equation}}

\newcommand{\vs}{\vspace{0.2cm}}

\numberwithin{equation}{section}
\newtheorem{Theorem}{Theorem}

\newtheorem{Definition}[Theorem]{Definition}
\newtheorem{Proposition}[Theorem]{Proposition}

\usepackage{tocloft} 

\usepackage{setspace, tocloft}

\allowdisplaybreaks

\begin{document}

\begin{center}
{\Large\bf On the existence of Killing fields in smooth 
\vs

spacetimes with a compact Cauchy horizon}

\vs\vs\vs

{\sc Mart\'in Reiris Ithurralde}

{mreiris@cmat.edu.uy}
\vs

{\sc Ignacio Bustamante Bianchi}

{ibustamante@cmat.edu.uy}

\vspace{.4cm}

{\it Centro de Matem\'atica, Universidad de la Rep\'ublica,} 

{\it Montevideo, Uruguay.}
\vs\vs

\end{center}

\begin{abstract}
We prove that the surface gravity of a compact non-degenerate Cauchy horizon in a smooth vacuum spacetime, can be normalized to a non-zero constant. This result, combined with a recent result by Oliver Petersen and Istv\'an R\'acz, end up proving the Isenberg-Moncrief conjecture on the existence of Killing fields, in the smooth differentiability class. The well known corollary of this, in accordance with the strong cosmic censorship conjecture, is that the presence of compact Cauchy horizons is a non-generic phenomenon. Though we work in $3+1$, the result is valid line by line in any $n+1$-dimensions.
\end{abstract}

\section{Introduction}
%
This article discusses the existence of Killing fields on {\it smooth} time-orientable vacuum $3+1$ - spacetimes $(\mathcal{M};g)$ having a compact connected Cauchy horizon $\mathcal{C}$. The horizon $\mathcal{C}$ is assumed to divide $\mathcal{M}$ into two connected regions, one of which is a globally hyperbolic spacetime having a smooth closed three-manifold as a Cauchy surface (for the notion of Cauchy horizon see \cite{MR757180}). Such $\mathcal{C}$ is known to be always a smooth \cite{MR3383325} totally geodesic null hypersurface of $\mathcal{M}$, ruled by null geodesics \cite{MR757180}. We will assume that $\mathcal{C}$ is {\it non-degenerate}, namely, that there is on it at least one future or past incomplete null  geodesic, (recall that an inextensible geodesic is {\it incomplete} if it has finite affine length). The ``future'' direction is relabeled if necessary so that at least an incomplete null geodesic points into it. 
\vs

The occurrence of compact Cauchy horizons is a rather peculiar and unique phenomenon of the General theory of Relativity whose conceptual and theoretical significance can be hardly overlooked. The spacetimes having Cauchy horizons contain regions that are not predictable from the initial data over the Cauchy surface and therefore display properties that conflict our intuition and the causal foundation of classical physics. Yet, according to the strong cosmic censorship conjecture, such peculiar occurrences should be in fact non-generic, (see for instance Geroch-Horowitz in \cite{MR544343}). Pointing into that direction, in 1983 James Isenberg and Vincent Moncrief started a series of seminal investigations to demonstrate that spacetimes with compact Cauchy horizons are indeed non-generic by proving first that they must always contain a particular type of Killing field \cite{MR709474}. Despite of the significant progress made in a wide class of situations \cite{MR709474}, \cite{MR2438980}, \cite{MR4066588}, the general existence's proof of a Killing field on smooth spacetimes remained elusive. In this article we prove this general conjecture using recent breakthroughs by Petersen and R\'acz in \cite{petersen2018symmetries} and by Petersen in \cite{petersen2019extension}. We discuss how this is done in the following lines.

The usual strategy to prove the existence of a Killing field requires sorting two difficulties. First, proving that there is a null nowhere-zero vector field $K$ over $\mathcal{C}$ such that,
\be\label{CONSTANTTEMP}
\nabla_{K}K=-K,
\ee
and second, proving that $K$ extends to a Killing field inside the spacetime $\mathcal{M}$. For technical reasons, the common approach to achieve these two steps required assuming that the spacetime was analytic. This was the prevalent hypothesis in \cite{MR709474} and in the sequel \cite{MR2438980} and \cite{MR4066588}, but it was also present in related black-hole contexts, in the seminal work of Hawking in \cite{MR293962}, in the work of Hollands, Ishibashi and Wald in \cite{MR2291793} or in that of Isenberg and Moncrief in \cite{MR2438980}. The advantage of this is that, by using the Einstein vacuum equations it is possible to expand in Taylor series an analytic $K$ satisfying (\ref{CONSTANTTEMP}) to obtain an actual Killing field inside the spacetime, (in this case the extension is only to the globally hyperbolic region). Thus, the second aforementioned step becomes feasible provided one can show the existence of that analytic $K$. Leaving aside the issue of finding it, the well known problem of working in the analytic class, clearly pointed out in \cite{MR709474}, is that analytic spacetimes are already non-generic inside the smooth class. Hence, if the aim is to prove that compact Cauchy horizons are non-generic, then one must necessarily remove this technical assumption. A significant progress in that direction was recently achieved by Petersen and R\'acz in \cite{petersen2018symmetries} and by Petersen in \cite{petersen2019extension}, that established the existence of a Killing field on smooth spacetimes provided there is a smooth $K$ on $\mathcal{C}$ satisfying (\ref{CONSTANTTEMP}), (as a result the extension is to both sides of $\mathcal{C}$). In turn, these works were based upon the breakthroughs by Petersen in \cite{petersen2018wave} on the wave equations with initial data on a Cauchy horizon. This article is devoted to prove the existence of such smooth $K$. The Isenberg-Moncrief conjecture in the smooth class then follows as a corollary.

We proceed to explain how to carry over the task. We claim first that if (\ref{CONSTANTTEMP}) holds then, at every $p$, $K(p)$ is necessarily the only null vector for which the affine length of the inextensible null geodesic starting from $p$ with velocity $K(p)$ is equal to one (see \cite{MR4066588}). This is seen as follows. Let $\gamma(s)$ be the geodesic with $\gamma(0)=p$ and $\gamma'(0)=K(p)$. Then $\gamma'(s)=f(s)K(\gamma(s))$ for some $f(s)$. Thus (\ref{CONSTANTTEMP}) implies $f'-f^{2}=0$ and $\gamma'(0)=K(p)$ implies $f(0)=1$. Hence $f(s)=1/(1-s)$, proving, as wished, that the affine length of $\gamma$ is one. Suppose now that all future pointing null geodesics on $\mathcal{C}$ are incomplete. Then one could define a {\it candidate vector field} $\tilde{K}$ (candidate to be the required smooth $K$) by the same intrinsic property just mentioned and that would satisfy $K$ if it were to exist: so let $\tilde{K}(p)$ be the only null vector at $p$ such that the affine length of the inextensible null geodesic starting from $p$ with velocity $\tilde{K}(p)$ is equal to one. This candidate vector field $\tilde{K}$ may fail to be smooth but (\ref{CONSTANTTEMP}) will still hold. Therefore, if $\tilde{K}$ is shown to be smooth then it will be the required vector field $K$. 

Based on the discussion above, the Isenberg-Moncrief conjecture in the smooth class will follow as a corollary of the next theorem.
\begin{Theorem}[Main Theorem]\label{MAINTHEO} Let $\mathcal{C}$ be a compact non-degenerate Cauchy horizon on a time orientable smooth spacetime. Then, all the future null geodesics of $\mathcal{C}$ have finite affine length. Furthermore, the candidate vector field is smooth.
\end{Theorem}

Note that if $Z$ is any smooth future null and nowhere zero vector field defined on an open set $U$, and $\psi^{-1}:U\rightarrow \mathbb{R}^{3}$ is a smooth local chart, then the candidate vector field adopts the local presentation,
\be
\tilde{K}(\psi(x,y,z))=L(x,y,z)Z(\psi(x,y,z)),
\ee
where $L(x,y,z)$ is the affine length of the null geodesic starting from $\psi(x,y,z)$ with velocity $Z(\psi(x,y,z))$. If $L(x,y,z)$ is shown to be smooth, then $\tilde{K}$ will be smooth on $U$. This local expression will be used later.
\vs

Before passing to the next section and aimed to orient the reader, let us say a few words about the proof of Theorem \ref{MAINTHEO} and the structure of the article. 

The proof of Theorem \ref{MAINTHEO} relies upon a particular implementation of the Isenberg-Moncrief ``ribbon argument", (for a discussion see \cite{MR4066588} and references therein). The ribbon argument has been used often in this and related problems, so it seems useful to discuss briefly here what is new about it in this article. 

In general, by using a ribbon argument one can link the affine length of two null future geodesics, and this permits suitable explicit local presentations of the candidate vector field. In our case such presentation is contained in equations (\ref{CANDIDATEEXP}) and (\ref{AFFINEL}). In simple terms, a ribbon argument relies always on exploiting the following crucial fact: for any null nowhere zero vector field $Z$ tangent to $\mathcal{C}$ and defined on an open set $U$, the one-form $\omega_{Z}$ on $U$ given by,
\be
\nabla_{Y}Z=:\omega_{Z}(Y)Z,
\ee
is {\it null-closed}, that is,
\be
d\omega_{Z}(Z,Y)=0,
\ee
for any $Y$. This fact is due to the vacuum Einstein equations and a proof of it in Gaussian null coordinates can be found in \cite{MR4066588}, (intrinsic proofs can be given too, \footnote{We are indebted to Oliver Petersen for showing us an unpublished intrinsic calculation.}). Hence, Stokes's theorem implies that the integral of $\omega_{Z}$ over the boundary of any compact two-manifold $\mathcal{R}$ that is embedded in $U$ and with $Z$ tangent to it, is zero. 
Now, the ribbon arguments use as $\mathcal{R}$ rectangles with two of its sides being null curves and the other two sides being curves transversal to the null directions. However, neither $Z$ nor the rectangles $\mathcal{R}$ admit canonical choices, and different selections may result in different implementations, hence in different presentations of the candidate vector field. This is a central point because certain presentations are better than others when it comes to prove smoothness. In this article the rectangles $\mathcal{R}$ are constructed using the notions of horizontal geodesic and horizontal parallel transport that depend upon fixing any smooth distribution of two-planes transversal to the null directions and that we introduce in Section \ref{S1}. Briefly, horizontal geodesics are curves whose velocity field is horizontally parallel and is tangent to the two-planes of the distribution. Now, fixed a null future and nowhere zero smooth vector field $X$ on $\mathcal{C}$, we construct the rectangles $\mathcal{R}$ by transporting an initial horizontal geodesic segment $\gamma:[0,a]\rightarrow \mathcal{C}$ along the null curve following $X$ and starting at $\gamma(0)$. The vector field $Z$ to be used is essentially the velocity field obtained while transporting the initial horizontal geodesic, and that turns out to be null and non-zero thanks to Proposition \ref{TOLO}, (as a matter of fact this is the only relevant proposition that we will use). It could happen that after such transporting process the rectangles obtained are just immersed and not embedded. Nonetheless it will be clear that this is of no importance and that the ribbon argument applies as well. As it will turn out the local presentation of the candidate vector field thus obtained and that is contained in equations (\ref{CANDIDATEEXP}) and (\ref{AFFINEL}), will be easily proved to be smooth.

The notions of ``horizontal geometry'', including the notion of horizontal geodesic and horizontal exponential map are discussed in Section \ref{S1}. The main theorem is treated in the Section \ref{S2} and in Section \ref{HGEO} we prove the couple of propositions stated in Section \ref{S1}, including Proposition \ref{TOLO}.
\vs

\noindent {\bf Aknowledgments}. We would like to thank Oliver Petersen and Piotr Chrusciel for their tight comments on the original manuscript.

\section{The horizontal exponential map}\label{S1}
Let $\mathfrak{p}:T\mathcal{C}\rightarrow \mathcal{C}$ be the tangent bundle of $\mathcal{C}$, (points in $T\mathcal{C}$ are denoted as usual by $(p,v)$, with $\mathfrak{p}(p,v)=p$). Let $\mathfrak{p}:N\rightarrow \mathcal{C}$ be the vector-bundle of null vectors tangent to $\mathcal{C}$. We call $N$ the {\it null bundle}. Let $\mathfrak{p}:H\rightarrow \mathcal{C}$ be any smooth distribution of two-planes in $T\mathcal{C}$, that we think as a vector bundle with two-dimensional fibers, such that $T\mathcal{C}=N\oplus H$. Having chosen $H$, we call it the {\it horizontal bundle}. The fibers of $N$ and $H$ over $p$ will be denoted by $N(p)$ and $H(p)$ respectively. Let $\pi:T\mathcal{C}=N\oplus H\rightarrow H$ be the natural projection (i.e. if $u=v\oplus w$ with $u\in T_{p}\mathcal{C}$, $v\in N(p)$ and $w\in H(p)$ then $\pi(p,u)=(p,w)$). 

A smooth vector field $Y$ on $\mathcal{C}$ is said to be {\it horizontal} if $Y(p)\in H(p)$ for all $p\in \mathcal{C}$ (i.e. $Y$ is a smooth section of $H$). Let $\nabla$ be the space-time covariant derivative restricted to $\mathcal{C}$ (recall that $\mathcal{C}$ is totally geodesic). Define the {\it horizontal covariant derivative} $D$ on $H$ as follows: if $X$ is a vector field on $\mathcal{C}$ and $Y$ is a horizontal vector field, then,
\be
D_{X}Y:=\pi(\nabla_{X}Y).
\ee
This horizontal covariant derivative defines {\it horizontal parallel fields} over curves in the usual manner: a horizontal vector field $V:(a,b)\rightarrow H$ over a curve $\gamma:(a,b)\rightarrow \mathcal{C}$ (that is $V(s)\in H(\gamma(s))$ for all $s\in (a,b)$), is {\it parallel} iff $D_{\gamma'}V=0$. Given $\gamma: [a,b]\rightarrow \mathcal{C}$ and a vector $v\in H(\gamma(a))$ one can always {\it parallel transport} $v$ along $\gamma$ obtaining thus a horizontal parallel field $V$ over $\gamma$ with $V(a)=v$. Observe that $D$ is compatible with the spacetime metric $g$ restricted to $H$, namely if $Y$ and $Z$ are horizontal vector fields, and $X$ is a vector field on $\mathcal{C}$ then, 
\begin{align}
X(g(Y,Z)) & =g(\nabla_{X}Y,Z)+g(Y,\nabla_{X}Z)=g(\pi(\nabla_{X}Y,Z)+g(Y,\pi(\nabla_{X}Z))\\ 
& =g(D_{X}Y,Z)+g(Y,D_{X}Z).
\end{align}
We assume from now on that $H$ is endowed with the metric $g$.

A curve $\gamma:(a,b)\rightarrow \mathcal{C}$ is said to be a {\it horizontal geodesic} if $\gamma'(s)\in H(\gamma(s))$ for all $s\in (a,b)$ and $D_{\gamma'}\gamma'=0$. The following basic proposition on the existence of horizontal geodesics will be proved in Section \ref{HGEO}.
\begin{Proposition}[Existence and uniqueness]\label{EXISTENCE2}
Given $p\in \mathcal{C}$ and $v\in H(p)$, there is a unique horizontal geodesic $\gamma:(-\infty,\infty)\rightarrow \mathcal{C}$ with $\gamma(0)=p$ and $\gamma'(0)=v$. 
\end{Proposition}

Then we define the horizontal exponential map, in the usual manner.
\begin{Definition} The {\it horizontal exponential map}, is the map $\overline{\exp}:H\rightarrow \mathcal{C}$ defined as,
\be
\overline{\exp}\,(p,v)=\gamma(1),
\ee
where $\gamma(s)$ is the unique horizontal geodesic with $\gamma(0)=p$ and $\gamma'(0)=v$. 
\end{Definition}
The map $\overline{\exp}$ will be of course smooth (it comes after solving a smooth ODE). The next proposition states the only crucial (unsubstitutable) property of the horizontal exponential map that we will need during the proof of the main theorem. Before it, define a curve $\alpha:[a,b]\rightarrow \mathcal{C}$ to be {\it null} if $\alpha'(s)\in N(\alpha(s))$ for all $s\in [a,b]$.
\begin{Proposition}[Transport of horizontal geodesics]\label{TOLO} Let $p\in \mathcal{C}$ and $v\in H(p)$, $v\neq 0$. Let $\alpha:[a,b]\rightarrow \mathcal{C}$ be a null curve with nowhere zero velocity and $\alpha(a)=p$, and let $V:[a,b]\rightarrow \mathcal{C}$ be the parallel transport of $v$ along $\alpha$. Then the curve $\beta:[a,b]\rightarrow \mathcal{C}$ given by $\beta(s)=\overline{\exp}\,({\alpha(s)},V(s))$ is a null curve with nowhere zero velocity. 
\end{Proposition}

We prove this proposition also in Section \ref{HGEO}.

\section{Proof of Theorem \ref{MAINTHEO}}\label{S2}

Before going into the proof we mention two important preliminary facts that will be used.
\begin{enumerate}
\item[(I)] First, and as mentioned in the introduction, any null vector field $Z$ on an open set $U$ of $\mathcal{C}$ gives rise to a one-form $\omega_{Z}$ over $U$ defined by, 
\be
\omega_{Z}(Y)Z=\nabla_{Y}Z,
\ee
for any $Y$ vector field on $U$. The form $\omega_{Z}$ has the central property that its exterior derivative is {\it null} in the sense that,
\be
d\omega_{Z}(Z,Y)=0, 
\ee
for any $Y$ vector field on $U$. As commented earlier, this is the crucial property allowing the Isenberg-Moncrief ``ribbon'' argument and will be used fundamentally. 

Note finally that if, $Z=fX$ then,
\be\label{CHANGEOFGAUGE}
\omega_{Z}(Z)=\frac{Z(f)}{f}+\omega_{X}(Z).
\ee

\item[(II)] Second, we mention how to compute the affine length of a null geodesic from a non-affine parametrization of it. Let $\gamma(s)$, $\gamma:[0,L)\rightarrow \mathcal{C}$ be an inextensible null geodesic. The affine length $L$ being finite or infinite. Let $s(\rho):[\rho_{0},\infty)\rightarrow [0,L)$ be a smooth change of parameter so that now $\gamma(s(\rho))$ is a null geodesic possibly with a non-affine parameterization. Letting $\gamma'=d\gamma/d\rho$ and $s'=ds/d\rho$, we compute,
\be\label{OMEGADEF}
\nabla_{\gamma'}\gamma'=(\frac{s''}{s'}) \gamma'=:\omega(\rho)\gamma',
\ee
Hence, having the expression for $\omega(\rho)$, the affine-length $L$ is computed by, 
\be\label{LENGTH}
L=s'(0)\int_{\rho_{0}}^{\infty}e^{\int_{\rho_{0}}^{\rho}\omega(\lambda)d\lambda}d\rho.
\ee
We deduce then, that in order to have affine length $L$ equal to one, it is necessary and sufficient to have,
\be
\frac{1}{s'(0)}=\int_{\rho_{0}}^{\infty}e^{\int_{\rho_{0}}^{\rho}\omega(\lambda)d\lambda}d\rho<\infty.
\ee 
This is equivalent to have $\int_{\rho_{0}}^{\infty}e^{\int_{\rho_{0}}^{\rho}\omega(\lambda)d\lambda}d\rho<\infty$ and to start the geodesic $\gamma(s)$ at $\gamma(0)$ with velocity equal to,
\be
\frac{d\gamma}{ds}\bigg|_{s=0}=(\int_{\rho_{0}}^{\infty}e^{\int_{\rho_{0}}^{\rho}\omega(\lambda)d\lambda}d\rho)\frac{d \gamma}{d\rho}\bigg|_{\rho=\rho_{0}}.
\ee
We will use this expression to given an explicit (local) presentation of the candidate vector field, that will be proved to be smooth.

Finally, note that if $X$ is a nowhere zero null vector field and $\gamma'(\rho)=f(\rho)X(\gamma(\rho))$ for some smooth $f(\rho)>0$, then, 
\be\label{OMEGA}
\omega=\frac{f'}{f}+\omega_{X}(\gamma'),
\ee
(compare this with (\ref{CHANGEOFGAUGE})).
\end{enumerate}
\vs

We introduce now additional notation.

During the proof, $B(0,r)\subset \mathbb{R}^{2}$ will denote the open ball of radius $r>0$ and centered at the origin. In $\mathbb{R}^{2}$ we use coordinates $(x,y)$.

It turns out that in order to parameterize incomplete geodesics in a controlled fashion it will be convenient to fix an auxiliary smooth nowhere zero future null vector field $X$. We fix such $X$ from now on and reserve the letter $X$ for it. Let $X^{*}$ be the one form such that $X^{*}(X)=1$ and $X^{*}(Y)=0$ for any horizontal vector field $Y$. The form $X^{*}$ is clearly smooth. Let $\varphi:\mathcal{C}\times (-\infty,\infty)\rightarrow \mathcal{C}$ be the smooth flow defined by $X$, namely, $\varphi(p,z)$ is the solution to the ODE,
\be
\frac{d\varphi (p,z)}{dz}=X(\varphi(p,z)),\quad \varphi(p,0)=p, 
\ee
hence with $z$ the parameter of the integral curves of $X$.

Finally we let $\mathfrak{p}: E\rightarrow \mathcal{C}$ be the bundle of orthonormal frames of $H$. Points in $E$ are denoted by $(p,\{e_{1},e_{2}\})$ with $\{e_{1},e_{2}\}$ an orthonormal basis (frame) of $H(p)$.  Given $(p,\{e_{1},e_{2}\})$ we denote by $\{e_{1}(z),e_{2}(z)\}$ to the horizontal parallel transport of $\{e_{1},e_{2}\}$ along the null curve $z\rightarrow \varphi(p,z)$. 

\begin{proof}[Proof of Theorem \ref{MAINTHEO}] We begin explaining the main arguments of the proof. The proof is divided in two obvious consecutive steps: (A) proving that all future null geodesics have finite affine length, (B) proving that the candidate Killing vector field $\tilde{K}$ is smooth. For (A) it will be enough to show that: (A') there are uniform $\epsilon>0$ and $\delta>0$ such that, if a future null geodesic from a point $p_{0}$ is incomplete, then all future null geodesics starting at any point in the uniform neighbourhood,
\be\label{UDEF}
U(p_{0},\epsilon,\delta)=\{\overline{\exp}(\varphi(p_{0},z),xe^{0}_{1}(z)+ye^{0}_{2}(z)):\ z^{2}<\delta^{2},\ x^{2}+y^{2}<\epsilon^{2}\},
\ee 
are also incomplete, (above $\{e_{1}^{0},e_{2}^{0}\}$ is any frame in $E(p_{0})$). What is important here is that $\epsilon$ and $\delta$ are independent on $p_{0}$. Indeed,  if we show (A') then the set of points with a null future incomplete geodesic will be open and closed, and thus (A) will follow from the connectivity of $\mathcal{C}$. Now, to prove (A'), but also for the proof of the step (B), it will play a simple but important role the map,
\be
(x,y,z)\rightarrow \overline{\exp}(\varphi(p_{0},z),xe_{1}^{0}(z)+ye_{2}^{0}(z)),
\ee
that we used to define $U(p_{0},\epsilon,\delta)$ in (\ref{UDEF}). Let us define it precisely in the next lines and inspect its properties. We will end up explaining the argument behind the proof of (B). 

Given $p_{0}$ and $\{e_{1}^{0},e_{2}^{0}\}\in E(p_{0})$, and given $\epsilon>0$ and $\delta>0$, define 
\be
\psi:B(0,\epsilon)\times (-\delta,\infty)\subset{\mathbb{R}^{3}}\rightarrow \mathcal{C},
\ee
as,
\be
\psi(x,y,z)=\overline{\exp}\, (\varphi(p_{0},z),xe^{0}_{1}(z)+ye^{0}_{2}(z)).
\ee 
Now, it is not difficult to prove and we will do later, that if $\epsilon>0$ and $\delta>0$ are small enough, then  for any $z_{1}\geq 0$ the restriction of $\psi$ to $B(0,\epsilon)\times (z_{1}-\delta,z_{1}+\delta)$ is an embedding. So let us assume such $\epsilon$ and $\delta$ for the rest of the discussion. By Proposition \ref{TOLO}, the curves,
\be
z\rightarrow \psi(x,y,z),
\ee
are null with non-zero velocity, and hence are null geodesics with $z$ a non-necessarily affine parameter (this is the only place where Proposition \ref{TOLO} is used). Therefore, their affine length can be calculated as was explained in (II). To make that explicit define $f(x,y,z)>0$ by, 
\be
d_{(x,y,z)}\psi (\partial_{z})=:f(x,y,z)X(\psi(x,y,z)), 
\ee
and then define the one-form $\omega^{*}$ over $B(0,\epsilon)\times (-\delta,\infty)\subset \mathbb{R}^{3}$ by,
\be
\omega^{*} := \frac{df}{f}+\psi^{*}\omega_{X},
\ee
which is null-closed in the sense that $d\omega^{*}(\partial_{z}, -)=0$ indeed by virtue of (I):
\begin{align}
d\omega^{*}(\partial_{z},Y)= & d(d\ln f) + d\psi^{*}\omega_{X}(\partial_{z},Y)\\
= & d\omega_{X}(\psi_{*}(\partial_{z}),\psi_{*}(Y))=0,
\end{align}
where in the last step we used that $\psi_{*}(\partial_{z})$ is null. (Just in passing, $\omega^{*}$ is the pull-back of the form $\omega_{Z}$ defined by the (just) local field $Z:=d\psi(\partial_{z})$, see (I)). Thus, using (\ref{OMEGA}) we obtain 
\be
\omega=\omega^{*}(\partial_{z}),
\ee
(see definition of $\omega$ in (\ref{OMEGADEF})) and then using (\ref{LENGTH}) we find that the affine length $L(x,y,z)$ of the future null geodesic starting from $\psi(x,y,z)$ with velocity $Z(\psi(z,y,z))$ takes the expression,
\be\label{L1}
L(x,y,z)=\int_{z}^{\infty}e^{\int_{z}^{\rho}\omega(x,y,\lambda)d\lambda}d\rho.
\ee
With this expression at hand, the goal of (A') is to prove that, if $L(0,0,0)<\infty$ then $L(x,y,z)<\infty$ for all $(x,y,z)$ with $x^{2}+y^{2}<\epsilon$ and $z^{2}<\delta^{2}$, whereas the goal of (B) is to prove that, once (A') is done, the following presentation of the candidate vector field,
\be\label{CANDIDATEEXP}
\tilde{K}(\psi(x,y,z))=L(x,y,z)Z(\psi(x,y,z)),
\ee
is smooth as a function of the smooth local coordinates $(x,y,z)$. Note that as we are restricting $(x,y,z)$ to $B(0,\epsilon)\times (-\delta,\delta)$ where $\psi$ is an embedding, the vector field $Z(\psi(x,y,z))$ is well defined and smooth over the patch $\psi(B(0,\epsilon)\times (-\delta,\delta))$. Thus, to achieve (B) we need to show that: (B') $L(x,y,z)$ is smooth. To prove (A') and (B') we need to link somehow $L(x,y,z)$ to $L(0,0,0)$. As was explained in the introduction, linking $L(x,y,z)$ to $L(0,0,0)$ is what the ribbon argument does. We explain how it works in the following lines. 

As $d\omega^{*}$ is null (i.e. $d\omega^{*}(\partial_{z},-)=0$), Stokes theorem shows that the integral of $\omega^{*}$ over the border of the rectangle $\mathcal{R}$ in $\mathbb{R}^{3}$ with vertices $(0,0,z)$, $(0,0,\rho)$, $(x,y,\rho)$ and $(x,y,z)$, is zero (note that $\partial_{z}$ is tangent to $\mathcal{R}$). We write this identity as,
\be\label{IDENTIT}
\int_{z}^{\rho}\omega(x,y,\lambda)d\lambda = S(x,y,\rho)-S(x,y,z)+\int_{z}^{\rho}\omega(0,0,\lambda)d\lambda,
\ee
where $S(x,y,z)$ is the integral of $\omega^{*}$ along the segment from $(0,0,z)$ to $(x,y,z)$ and $S(x,y,\rho)$ is the one from $(0,0,\rho)$ to $(x,y,\rho)$. Then (\ref{IDENTIT}) transforms (\ref{L1}) into,
\be\label{AFFINEL}
L(x,y,z)=\int_{z}^{\infty}e^{S(x,y,\rho)-S(x,y,z)}e^{\int_{z}^{\rho}\omega(0,0,\lambda)d\lambda}d\rho.
\ee
The important point here is that positive integrand $e^{\int_{z}^{\rho}\omega(0,0,\lambda)d\lambda}$ in this integral is integrable by virtue of $L(0,0,0)<\infty$. That is the desired link between $L(0,0,0)$ and $L(x,y,z)$. 

Now, we claim that (A') and (B') follow after proving that the function $S(x,y,z)$ and all the partial derivatives of it of any given order are bounded all over $B(0,\epsilon)\times (-\delta,\infty)$, (the bounds may depend on the order). Note also that the function $\omega(0,0,z)=\omega^{*}(\partial_{z})(0,0,z)=X^{*}(\nabla_{X}X)(\varphi(p_{0},z))$ that appears also in (\ref{AFFINEL}) is uniformly bounded and so are all of its derivatives. 

In fact, assuming that property for $S(x,y,z)$ that we will show later, then (A') follows from the bound,
\be
L(x,y,z)\leq e^{2\|S\|_{L^{\infty}}}L(0,0,0),
\ee
whereas (B') follows by a simple induction in the order of the derivatives after checking that one can use Leibniz's rule for differentiation under an improper integral sign. Let us set the induction precisely. Define $\mathcal{B}$ to be the space of functions, 
\be
F(x,y,z):B(0,\epsilon)\times (-\delta,\infty)\rightarrow \mathbb{R},
\ee
and, 
\be
G(x,y,z,\rho): B(0,\epsilon)\times \{(z,\rho)\in \mathbb{R}^{2}: \rho\geq z\geq -\delta\}\rightarrow \mathbb{R},
\ee
that are bounded and have all the derivatives of a given order also bounded. Then, the induction to prove (B') is set as follows: 

If for all multi-index $I=(i_{1},i_{2},i_{3})$, with $|I|=i_{1}+i_{2}+i_{3}= k$, we have,
\be\label{1}
\frac{\partial^{|I|} L(x,y,z)}{\partial_{x}^{i_{1}}\partial_{y}^{i_{2}}\partial_{z}^{i_{3}}} = F_{I}(x,y,z)+\int_{z}^{\infty}G_{I}(x,y,z,\rho)e^{\int_{z}^{\rho}\omega(0,0,\lambda)d\lambda}d\rho,
\ee
for some $F_{I}$ and $G_{I}$ in $\mathcal{B}$, then for all multi-index $I'=(i'_{1},i'_{2},i'_{3})$ with $|I'|=i'_{1}+i'_{2}+i'_{3}= k+1$, we have,
\be\label{2}
\frac{\partial^{|I'|} L(x,y,z)}{\partial_{x}^{i'_{1}}\partial_{y}^{i'_{2}}\partial_{z}^{i'_{3}}} = F_{I'}(x,y,z)+\int_{z}^{\infty}G_{I'}(x,y,z,\rho)e^{\int_{z}^{\rho}\omega(0,0,\lambda)d\lambda}d\rho,
\ee
for some $F_{I'}$ and $G_{I'}$ in $\mathcal{B}$. 

Note that for $k=0$, equation (\ref{L1}) has the form (\ref{1}) with $F_{0}=0$ and $G_{0}=e^{S(x,y,\rho)-S(x,y,z)}$. To prove the inductive step we need to calculate the derivatives carefully. First, the derivative of (\ref{1}) with respect to $x$ is,
\be\label{3}
\frac{\partial}{\partial x}\frac{\partial^{|I|} L(x,y,z)}{\partial_{x}^{i_{1}}\partial_{y}^{i_{2}}\partial_{z}^{i_{3}}} =\frac{\partial F_{I}(x,y,z)}{\partial x}+\int_{z}^{\infty}\frac{\partial G_{I}(x,y,z,\rho)}{\partial x}e^{\int_{z}^{\rho}\omega(0,0,\lambda)d\lambda}d\rho,
\ee
where the differentiation inside the integral is permitted by virtue of the fact that $\partial_{x} (G_{I}e^{\int_{z}^{\rho}\omega(0,0,z)d\lambda})$ is continuous but also bounded by the integrable function of $\rho$, $Ce^{\int_{z}^{\rho}\omega(0,0,z)d\lambda}$. This is a pretty straightforward fact\footnote{The precise statement is: if $f(x,t)$ and $\partial_{x}f(x,t)$ are continuous, $|f(x,t)|\leq g_{1}(t)$ and $|\partial_{x}f(x,t)|\leq g_{2}(t)$ with $\int_{t_{0}}^{\infty}g_{1}(\tau)d\tau<\infty$ and $\int_{t_{0}}^{\infty}g_{2}(\tau)d\tau<\infty$, then the function $x\rightarrow \int_{t_{0}}f(x,\tau)d\tau$ is differentiable and its derivative is equal to $\int_{t_{0}}^{\infty}\partial_{x}f(x,\tau)d\tau$.}, that can be found for instance in Theorem 15 in Chp 8 of \cite{MR1123269}. A similar calculation holds for the derivative with respect to $y$, whereas the derivative with respect to $z$ is directly,
\begin{align}\label{4}
\frac{\partial}{\partial z} & \frac{\partial^{|I|} L(x,y,z)}{\partial_{x}^{i_{1}}\partial_{y}^{i_{2}}\partial_{z}^{i_{3}}} =  \frac{\partial F_{I}(x,y,z)}{\partial z} - G_{I}(x,y,z,z) \\
& \hspace{1.7cm}+ \int_{z}^{\infty}\big(\frac{\partial G_{I}(x,y,z,\rho)}{\partial z} - G_{I}(x,y,z,\rho)\omega(0,0,z)\big)e^{\int_{z}^{\rho}\omega(0,0,\lambda)d\lambda}d\rho.
\end{align}
The proof of the inductive step thus follows. This would finish the proof of (A') and (B') and hence so of (A) and (B).
\vs

We pass now to prove the claims that were left to be proved, namely (i) to show the existence of $\epsilon>0$ and $\delta>0$ such that $\psi:B(0,\epsilon)\times (-\delta+z_{1},\delta+z_{1})\rightarrow \mathcal{C}$ is an embedding for any $z_{1}\geq 0$, and (ii) show that $S(x,y,z):B(0,\epsilon)\times (-\delta,\infty)\rightarrow \mathbb{R}$ as well as any of its derivatives are bounded (again, the bounds may depend on the order of the derivative). We prove (i) first and then (ii). Both are basically the result of compactness. Before the proof we make some analysis.

We define first a smooth map $\phi$ from $E\times B(0,1) \times (-1,1)$ into $\mathcal{C}$, (recall $E$ is the frame bundle of $H$). Points in $E\times B(0,1) \times (-1,1)$ are denoted by $((p,\{e_{1},e_{2}\}),(x,y),z)$. The map $\phi$ is given by,
\be\label{MAINDEF}
\phi((p,\{e_{1},e_{2}\}),(x,y),z):=\overline{\exp}\, (\varphi(p,z),xe_{1}(z)+ye_{2}(z)).
\ee 
At any point $P=((p,\{e_{1},e_{2}\}),(0,0),0)$ we compute,
\be\label{DIFF}
d_{P}\phi(\partial_{x})=e_{1},\quad d_{P}\phi(\partial_{y})=e_{2},\quad d_{P}\phi(\partial_{z})=X.
\ee
Hence, at any $P\in E$ there is $0<\epsilon<1$ and $0<\delta<1$ such that the map $\phi$ restricted to $\{P\}\times B(0,3\epsilon)\times (-3\delta,3\delta)$ is an embedding. By continuity there is a neighbourhood $U_{P}$ such that at any $P'\in U_{P}$ the map $\phi$ restricted to $\{P'\}\times B(0,2\epsilon)\times (-2\delta,2\delta)$ is an embedding. As $E$ is compact then there are uniform $0<\epsilon<1$ and $0<\delta<1$, such that at any $P\in E$ the map $\phi$ restricted to $\{P\}\times B(0,2\epsilon)\times (-2\delta,2\delta)$ is an embedding. Also, taking into account the third equation in (\ref{DIFF}), that we rewrite as $X^{*}(d_{P}\phi(\partial_{z}))=1$ for all $P\in E\times \{(0,0)\}\times \{0\}$, we can decrease $\epsilon$ and $\delta$ if necessary so that, in addition, 
\be
X^{*}(d\phi(\partial_{z}))\geq \frac{1}{2},
\ee
all over $E\times B(0,2\epsilon)\times (-2\delta,2\delta)$. We fix such $\epsilon$ and $\delta$ from now on. 

Consider now the following four smooth functions from $E\times B(0,2\epsilon)\times (-2\delta,2\delta)$ into $\mathbb{R}$,
\be\label{FOURFUNCTIONS}
\ln (X^{*}(d\phi(\partial_{z}))),\quad \omega_{X}(d\phi(\partial_{x})),\quad \omega_{X}(d\phi(\partial_{y})),\ \ {\rm and}\ \ \omega_{X}(d\phi(\partial_{z})).
\ee
Trivially, the four of them are bounded functions when restricted to the compact set $C:=E\times \overline{B(0,\epsilon)}\times [-\delta,\delta]\subset E\times B(0,2\epsilon)\times (-2\delta,2\delta)$. The same of course holds true for any partial derivative of any order in $x,y,$ and $z$. We state this as follows,  
\be\label{MAINBOUNDS}
\| \frac{\partial^{|I|} h}{\partial x^{i_{1}}\partial y^{i_{2}}\partial z^{i_{3}}}\|_{L^{\infty}(C)}\leq c(|I|),
\ee
where $I$ a multi-index $I=(i_{1},i_{2},i_{3})$, $|I|=i_{1}+i_{2}+i_{3}$, and $h$ is any of the four functions (\ref{FOURFUNCTIONS}). We will see now that these trivial bounds are ultimately all the necessary bounds.

We are finally are in position to prove (i) and (ii). The basic observation is that for any $z_{1}\geq 0$ the map $\psi$ restricted to $B(0,\epsilon)\times (z_{1}-\delta,z_{1}+\delta)$ is ``equal'' to the map $\phi$ restricted to $\{(\varphi_{z_{1}}(p_{0}),\{e_{1}^{0}(z_{1}),e_{2}^{0}(z_{1})\})\}\times B(0,\epsilon)\times (-\delta,\delta)$, more precisely if we define,
\be
\chi:B(0,\epsilon)\times (-\delta+z_{1},\delta+z_{1})\rightarrow \{(\varphi_{z_{1}}(p_{0}),\{e_{1}^{0}(z_{1}),e_{2}^{0}(z_{1})\})\}\times B(0,\epsilon)\times (-\delta,\delta),
\ee
by,
\be
\chi(x,y,z)=\phi((\varphi_{z_{1}}(p_{0}),\{e_{1}^{0}(z_{1}),e_{2}^{0}(z_{1})\}),(x,y),z-z_{1}),
\ee
then, $\psi(x,y,z)=\phi(\chi(x,y,z))$. This shows (i), namely, that for any $z_{1}\geq 0$, the map $\psi$ restricted to $B(0,\epsilon)\times (z_{1}-\delta,z_{1}+\delta)$ is an embedding. To show (ii) we proceed as follows. First, as $d\chi(\partial_{x})=\partial_{x}$, $d\chi(\partial_{y})=\partial_{y}$ and $d\chi(\partial_{z})=\partial_{z}$, we deduce that,
\begin{gather}\label{EQUALFOURFUNCTIONS}
\ln (X^{*}(d\psi(\partial_{z})))\bigg|_{(x,y,z)}=\ln (X^{*}(d\phi(\partial_{z})))\bigg|_{\chi(x,y,z)},\\
\omega_{X}(d\psi(\partial_{x}))\bigg|_{(x,y,z)}=\omega_{X}(d\phi(\partial_{x}))\bigg|_{\chi(x,y,z)},\\ 
\omega_{X}(d\psi(\partial_{y}))\bigg|_{(x,y,z)}=\omega_{X}(d\phi(\partial_{y}))\bigg|_{\chi(x,y,z)},\\
\omega_{X}(d\psi(\partial_{z}))\bigg|_{(x,y,z)}=\omega_{X}(d\phi(\partial_{z}))\bigg|_{\chi(x,y,z)}.
\end{gather}
It follows then from this and from (\ref{MAINBOUNDS}) that, 
\be
\| \frac{\partial^{|I|} \bar{h}}{\partial x^{i_{1}}\partial y^{i_{2}}\partial z^{i_{3}}}\|_{L^{\infty}(B(0,\epsilon)\times (-\delta+z_{1},\delta+z_{1}))}\leq c(|I|),
\ee
where $I$ is the multi-index $I=(i_{1},i_{2},i_{3})$, $|I|=i_{1}+i_{2}+i_{3}$, $c(|I|)$ are the same constants as in (\ref{MAINBOUNDS}) and $\bar{h}$ is now any of the four functions,
\begin{gather}\label{DOUBLEFOURFUNCTIONS}
f:=\ln (X^{*}(d\psi(\partial_{z}))),\\ 
\varpi_{x}:=\omega_{X}(d\psi(\partial_{x})),\quad \varpi_{y}:=\omega_{X}(d\psi(\partial_{y})),\quad \varpi_{y}:=\omega_{X}(d\psi(\partial_{z})).
\end{gather}
Finally, as these estimates are valid for any $z_{1}\geq 0$, we obtain that the form,
\be
\omega^{*} = \frac{df}{f}+\psi^{*}\omega_{X} = \frac{\partial_{x}f}{f}dx+\frac{\partial_{y}f}{f}dy+\frac{\partial_{z}f}{f}dz+\varpi_{x}dx+\varpi_{y}dy+\varpi_{z}dx,
\ee
is bounded and has all its derivatives of any order bounded over $B(0,\epsilon)\times (-\delta,\infty)\subset \mathbb{R}^{3}$. This directly proves (ii) namely that $S(x,y,z)$ and all its derivatives of any order are bounded, as wished.
\end{proof}

\section{Proof of Propositions \ref{EXISTENCE2} and \ref{TOLO}}\label{HGEO} 

In this section we prove Propositions \ref{EXISTENCE2} and \ref{TOLO}.

Recall that a curve $\gamma:(a,b)\rightarrow \mathcal{C}$ is a horizontal geodesic if for all $s\in (a,b)$, 
\be\label{HG}
\gamma'(s)\in H(\gamma(s))\quad {\rm and}\quad \pi(\nabla_{\gamma'}\gamma')(s)=0.
\ee

\begin{Proposition}\label{EXISTENCE} 
Let $p\in \mathcal{C}$ and $v\in H(p)$. Then, for $\epsilon>0$ sufficiently small, there exists a unique horizontal geodesic $\gamma:(-\epsilon,\epsilon)\rightarrow \mathcal{C}$ with $\gamma(0)=p$ and $\gamma'(0)=v$. 
\end{Proposition}
\begin{proof}[Proof of Proposition \ref{EXISTENCE}] We make a local calculation. Let $B$ be an embedded disc containing $p$ and transversal to the null directions. For $\delta>0$ sufficiently small, the restriction of $\varphi$ to $B\times (-\delta,\delta)$ is an embedding ($\varphi$ is again the flow generated by $X$, the vector field that we fixed in Section \ref{S2}). Let $U=\varphi(B\times (-\delta,\delta))$. The open set $U$ is foliated by the null orbits $\{\{\varphi(p,z):z\in (-\delta,\delta)\} : p\in B\}$. Let $V$ be the quotient of $U$, and note that obviously $V$ is diffeomorphic to $B$. Let $\xi:U\rightarrow V$ be the projection. Any function $f$ on $V$ lifts to a function $f^{*}$ on $U$ by: $f^{*}(p):=f(\xi(p))$. Also, any vector field $Y$ on $V$ lifts to a horizontal vector field $Y^{*}$ on $U$ by: $Y^{*}(p)\in H(p)$ and $d_{p}\xi (Y^{*}(p))=Y(\xi(p))$. Note that for any function $f$ and vector field $Y$ on $V$ we have, $Y^{*}(f^{*})=(Y(f))^{*}$. Also, note that $\pi([Y^{*},Z^{*}])=[Y,Z]^{*}$, (again $\pi$ is the horizontal projection, see Section \ref{S1}). Indeed, for any function $f$ on $V$ we have,
\begin{align}
\pi([Y^{*},Z^{*}])(f^{*}) & = [Y^{*},Z^{*}](f^{*}) = \\
& = Y^{*}(Z^{*}(f^{*}))-Z^{*}(Y^{*}(f^{*}))=(Y(Z(f))-Z(Y(f)))^{*}= \\
& = ([Y,Z](f))^{*}.
\end{align}

Let $h$ be the degenerate metric on the horizon $\mathcal{C}$. Such tensor is the restriction to $\mathcal{C}$ of the spacetime metric $g$. As $\mathcal{L}_{X} h = 0$, \footnote{This is because $\mathcal{L}_{X}h(Y,W)=g(\nabla_{Y}X,W)+g(\nabla_{W}X,Y)=g(\omega_{X}(Y)X,W)+g(\omega_{X}(W)X,Y)=0$.}, the metric $h$ on $U$ can be quotient to a metric $q$ on $V$. Note that $\langle Y^{*},Z^{*}\rangle = \langle Y,Z\rangle^{*}$, where with some abuse of notation, (that will be used below too), the first bracket corresponds to the degenerate metric $h$ ($\langle Y^{*},Z^{*} \rangle = h(Y^{*},Z^{*})$) and the second to the metric $q$ ($\langle Y,Z \rangle = q(Y,Z)$).

We show now that the covariant derivative on $V$ defined by,
\be\label{PDER}
{\mathcal{D}}_{Y}Z:=d\xi(\pi(\nabla_{Y^{*}}Z^{*}))=d\xi(D_{Y^{*}}Z^{*}),
\ee
is indeed the Levi-Civita connection of $q$ (in this formula $D$ is the horizontal covariant derivate on $H$, see Section \ref{S1}). Note that though $\pi(\nabla_{Y^{*}}Z^{*})$ is well defined as a vector field on $U$, it is not necessarily projectable to a vector field on $V$, so in principle (\ref{PDER}) may not be well defined. That it is indeed well defined will be clear in the following calculation. By the standard formula, for any vector field $W$ on $V$ we have,
\begin{align}
\label{CV1} \langle \pi(\nabla_{Y^{*}}Z^{*}),W^{*}\rangle = & \langle \nabla_{Y^{*}}Z^{*},W^{*}\rangle = \\ 
\label{CV2} = & \frac{1}{2}\big\{Z^{*} \langle Y^{*},W^{*}\rangle +Y^{*}\langle W^{*},Z^{*}\rangle -W^{*}\langle Y^{*},Z^{*}\rangle \\ 
\label{CV3} &- \langle [Z^{*},W^{*}],Y^{*}\rangle -\langle [Y^{*},W^{*}],Z^{*}\rangle -\langle [Z^{*},Y^{*}],W^{*}\rangle \big\}.
\end{align} 
Now,
\be
Z^{*} \langle Y^{*},W^{*}\rangle = Z^{*}\langle Y,W\rangle^{*}=(Z\langle Y,W\rangle)^{*},
\ee
and similarly for the other two terms in (\ref{CV2}). Also, 
\be
\langle [Z^{*},W^{*}],Y^{*}\rangle = \langle \pi([Z^{*},W^{*}]),Y^{*}\rangle = \langle [Z,W]^{*},Y^{*}\rangle = \langle [Z,W],Y\rangle^{*},
\ee
and similarly for the other terms in (\ref{CV3}). Putting all together we obtain,
\begin{align}
\label{CV12} \langle \pi(\nabla_{Y^{*}}Z^{*}),W^{*}\rangle = & \frac{1}{2}\{Z \langle Y,W\rangle +Y\langle W,Z\rangle -W\langle Y,Z\rangle \\ 
\label{CV32} &- \langle [Z,W],Y\rangle -\langle [Y,W],Z\rangle -\langle [Z,Y],W\rangle\}^{*} = \\
\label{CV4} = & \langle \nabla_{Y}Z,W\rangle^{*},
\end{align}
where on the left hand side of (\ref{CV12}) the covariant derivative is that of $g$ and on (\ref{CV4}) the covariant derivative is that of $q$. Thus $\mathcal{D}$ is the Levi-Civita connection of $q$.

Let now $\gamma(s)$ be a horizontal curve, namely $\gamma'(s)\in H(\gamma(s))$ for all $s$. Let $\alpha(s)=\xi(\gamma(s))$. Then, the calculation earlier shows that,
\be
d\xi (\pi (\nabla_{\gamma'}\gamma')) = \mathcal{D}_{\alpha'}\alpha'.
\ee
Hence, if $\gamma(s)$ is a horizontal geodesic on $U$, then $\alpha(s)$ is a geodesic on $V$. So if $\gamma:(-\epsilon,\epsilon)\rightarrow \mathcal{C}$ is a horizontal geodesic on $U$ then $\alpha(s)=\xi(\gamma(s))$ is a geodesic on $V$ with $\alpha(0)=\xi(p)$ and $\alpha'(0)=d\xi(v)$. Conversely, if $\alpha:(-\epsilon,\epsilon)\rightarrow V$ is a geodesic on $V$, with $\alpha(0)=\pi(p)$ and $\alpha'(0)=d\xi(v)$ then one can lift it to a unique horizontal curve $\gamma(s)$, with $\gamma(0)=p$, $\gamma'(0)=v$, that will be the horizontal geodesic we are looking for.
\end{proof}

Now note that $|\gamma'|^{2}{'}=\langle \gamma',\gamma'\rangle ' = 2\langle \nabla_{\gamma'}\gamma',\gamma'\rangle = 2\langle \pi(\nabla_{\gamma'}\gamma'),\gamma'\rangle =0$, and thus the norm of $\gamma'$ is constant. A standard argument using the compactness of $\mathcal{C}$ then shows that any horizontal geodesic $\gamma:(a,b)\rightarrow \mathcal{C}$ can be uniquely continued to a horizontal geodesic $\gamma:(-\infty,\infty)\rightarrow \mathcal{C}$, thus proving Proposition \ref{EXISTENCE2}.

Let us move now to prove Proposition \ref{TOLO}. Let us remain for some lines inside the context of the proof just made of Proposition \ref{EXISTENCE}. Let $p_{1}$ and $p_{2}$ be two points in $U$ projecting into the same point, $\xi(p_{1})=\xi(p_{2})$. Let $v_{1}\in H(p_{1})$ and $v_{2}\in H(p_{2})$ projecting into the same vector, $d\xi(v_{1})=d\xi(v_{2})$. Assume that the norms of $v_{1}$ and $v_{2}$ (which are necessarily equal) is small enough so that the horizontal geodesics $\gamma_{1}:[0,1]\rightarrow \mathcal{C}$ and $\gamma_{2}:[0,1]\rightarrow \mathcal{C}$ starting from $p_{1}$ and $p_{2}$ with velocities $v_{1}$ and $v_{2}$ respectively lie inside $U$. Then, $\alpha_{1}=\xi(\gamma_{1})$ and $\alpha_{2}=\xi(\gamma_{2})$ are geodesics of $V$ that have the same initial data and are thus equal. This shows that $\xi(\gamma_{1}(1))=\xi(\gamma_{2}(1))$, or, equivalently,
\be\label{JUSTIFICATION}
\xi(\overline{\exp}(p_{2},v_{2}))=\xi(\overline{\exp}(p_{1},v_{1})). 
\ee
Now, we claim that $v_{2}$ is the horizontal parallel transport of $v_{1}$ from $p_{1}=\varphi(p_{1},0)$ to $p_{2}=\varphi(p_{1},z_{2})$.
%
Indeed if we let $V(z)\in H(\varphi(p_{1},z))$ be the unique horizontal field over the null curve $z\rightarrow \varphi(p_{1},z)$ such that $d\xi(V(z))=d\xi(v_{1})$, and thus with $V(0)=v_{1}$ and $V(z_{2})=v_{2}$, 
then the claim follows by the computation.
\be
0 = \pi(\mathcal{L}_{X}V) = \pi(\nabla_{X}V-\nabla_{V}X)=D_{X}V-\pi(\omega_{X}(V)X)=D_{X}V.
\ee

In sum what we have shown is that, given $p\in \mathcal{C}$ and $v\in H(p)$, there are $\epsilon(p)>0$ and $\delta(p)>0$ such that, if we let $V(z)$ be the horizontal parallel transport of $v$ along the null curve $z\rightarrow \varphi(p,z)$, then the family of horizontal geodesics $\beta(z,s):[0,\epsilon]\times [-\delta,\delta]\rightarrow \mathcal{C}$, given as,
\be\label{FINALCURVE}
\beta(z,s)=\overline{\exp}(\varphi(p,z),sV(z)),
\ee
all project into the same geodesic in $V$, that is,
\be
\xi(\beta(z,s))=\xi(\beta(0,s)). 
\ee
Therefore the curves, $z\rightarrow \beta(z,s)$ are all null, and the horizontal fields over them, $z\rightarrow \partial_{s} \beta(z,s)$, are all horizontally parallel. Of course, if $\epsilon$ and $\delta$ are small enough then $\partial_{z}\beta(z,s)\neq 0$, for all $z\in [0,\epsilon]$ and $s\in [-\delta,\delta]$. Of course too one can chose $\epsilon(p)$ and $\delta(p)$ such that if $p'$ is sufficiently close to $p$, then the same holds with $\epsilon(p')=\epsilon(p)$ and $\delta(p')=\delta(p)$.

To prove Proposition \ref{TOLO} we will use what we know so far and make a simple continuity argument. Let $p\in \mathcal{C}$ and $v\in H(p)$, $v\neq 0$ but arbitrary. Let again $V(z)$ be the horizontal parallel transport of $v$ along the null curve $z\rightarrow \varphi(p,z)$, for all $z\in \mathbb{R}$. Let $\beta(z,s)=\overline{\exp}(\varphi(p,z),sV(z))$, for all $s\geq 0$. We will show that $\partial_{z}\beta(z,s)$ is null and different from zero for all $z\in \mathbb{R}$ and all $s\geq 0$. This is clearly enough to prove the proposition. Observe that $z=0$ doesn't play any particular role, so it is enough to prove that that $\partial_{z}\beta(0,s)$ is null and different from zero for all $s\geq 0$.

Let $s_{*}$ be the supremum of the $s_{1}>0$ for which there is $\epsilon=\epsilon(s_{1})>0$ such that for all $s\leq s_{1}$, we have: (i) the curves $z\rightarrow \beta(z,s)$ are null, and $\partial_{z}\beta(z,s)\neq 0$ for all $z\in [0,\epsilon]$, (ii) the fields $z\rightarrow \partial_{s}\beta(z,s)$ along the curves $z\rightarrow \beta(z,s)$, are horizontally parallel. By what was proved earlier we have $s_{*}>0$. If $s_{*}=\infty$ we are done. So assume $0<s_{*}<\infty$. 

Let $s_{1}=s_{*}-\delta$, for some $\delta>0$ that we will chose soon. Let $p_{1}=\beta(0,s_{1})$, $v_{1}=\partial_{s}\beta(0,s_{1})$ and let $\beta_{1}(z,s)=\overline{\exp}(\varphi(p_{1},z),(s-s_{1})V_{1}(z))$ where $V_{1}(z)$ is the horizontal parallel transport of $v_{1}$ along $z\rightarrow \varphi(p_{1},z)$. Then, it is direct that from (i) and (ii) and the fact that $s_{1}<s_{*}$, that there is a function $z_{1}(z):[0,\epsilon(s_{1})] \rightarrow \mathbb{R}$, with $z_{1}'(z)\neq 0$ for which we have $\beta(z,s)=\beta_{1}(z_{1}(z),(s-s_{1})V_{1}(z_{1}(z)))$, when $s_{1}\leq s<s_{*}+\delta$. 

But as shown earlier too, if $\delta>0$ is sufficiently small, there is $\epsilon(\delta)>0$, such that: (i') the curves $z\rightarrow \beta_{1}(z,s)$ are null with $\partial_{z}\beta_{1}(z,s)\neq 0$ for all $z\in [0,\epsilon(\delta)]$, (ii') the field $z\rightarrow \partial_{s}\beta_{1}(z,s)$ along the curves $z\rightarrow \beta_{1}(z,s)$, are horizontally parallel. 

It follows from the paragraphs above that for $z\in [0,\epsilon(s_{1})]$ and for $s$ in the interval $s_{1}=s_{*}-\delta<s<s_{*}+\delta$ we have $\partial_{z}\beta_{1}(z,s)\neq 0$, hence $\partial_{z}\beta(z,s)=\partial_{z}\beta_{1}(z_{1}(z),s-s_{1})z_{1}'(z)$ is null and different from zero, and, furthermore, $\partial_{s}\beta(z,s)=\partial_{s}\beta_{1}(z_{1}(z),s-s_{1})$ is horizontally parallel. We reach thus a contradiction, having assumed $s_{*}<\infty$. Thus $s_{*}=\infty$ and the Proposition \ref{TOLO} is proved. 

\bibliographystyle{plain}
\bibliography{Master}

\end{document}